\pdfoutput=1
\documentclass{article}
\usepackage{fullpage}
\usepackage{mathtools}
\usepackage{amsmath}
\usepackage{amsfonts}
\usepackage{amsthm}
\usepackage{algorithmic}
\usepackage[ruled]{algorithm2e}
\usepackage{bbm}
\usepackage{color}
\usepackage{enumitem}
\usepackage[dvipsnames]{xcolor}
\usepackage{thmtools}

\usepackage{hyperref}
\usepackage[capitalize,nameinlink]{cleveref}
\hypersetup{colorlinks,linkcolor={blue},citecolor={blue},urlcolor={red}} 

\newtheorem{theorem}{Theorem}[section]
\newtheorem{lemma}[theorem]{Lemma}

\theoremstyle{definition}
\newtheorem{definition}[theorem]{Definition}
\newtheorem{invariant}[theorem]{Invariant}

\newcommand{\whp}{\emph{whp}\xspace}
\newcommand{\eps}{\varepsilon}
\newcommand{\defn}[1]{\textbf{\emph{#1}}}
\newcommand{\etal}{et al.\xspace{}}

\newcommand{\core}{k}

\newcommand{\kest}{\hat{\core}}

\newcommand{\lf}{\psi}
\newcommand{\upexp}{(1+\lf)}
\newcommand{\upexpold}{(2+\lambda)(1+\lf)}

\newcommand{\hnup}{\widehat{\nup}}

\newcommand{\lcur}{r}

\newcommand{\degen}{d}

\newcommand{\release}{\textbf{releases}\xspace}

\newcommand{\multfactor}{\phi}

\newcommand{\addfactor}{\zeta}
\newcommand{\order}{J}

\newcommand{\rangeout}{\mathcal{Y}}
\newcommand{\adj}{\mathbf{a}}
\newcommand{\coren}{\eta}

\newcommand{\numlevels}{4\log^2 n}
\newcommand{\numgrouplevels}{2\log n}

\newcommand{\nup}{\mathcal{U}}

\newcommand{\alg}{\mathcal{A}}

\newcommand{\poly}{\text{poly}\xspace}
\newcommand{\kc}{$k$-core\xspace}

\crefname{theorem}{Theorem}{Theorems}
\crefname{lemma}{Lemma}{lemmas}
\Crefname{lemma}{Lemma}{Lemmas}
\Crefname{claim}{Claim}{Claims}
\crefname{invariant}{Invariant}{Invariants}
\Crefname{observation}{Observation}{Observations}
\Crefname{algorithm}{Algorithm}{Algorithms}
\Crefname{myalgctr}{Algorithm}{Algorithms}
\Crefname{challenge}{Challenge}{Challenges}
\crefalias{AlgoLine}{line}

\crefname{algocf}{alg.}{algs.}
\Crefname{algocf}{Algorithm}{Algorithms}

\title{Near-Optimal Differentially Private $k$-Core Decomposition}
\author{Laxman Dhulipala$^1$, George Z. Li$^1$, and Quanquan C. Liu$^2$}
\date{$^1$University of Maryland and $^2$Simons Institute at UC Berkeley}

\begin{document}

\maketitle

\begin{abstract}
    Recent work by Dhulipala et al.~\cite{DLRSSY22} initiated the study of the $k$-core decomposition problem under differential privacy via a connection between low round/depth distributed/parallel graph algorithms and private algorithms with small error bounds. They showed that one can output differentially private approximate $k$-core numbers, while only incurring a multiplicative error of $(2 +\eta)$  (for any constant $\eta >0$) and additive error of $\poly(\log(n))/\eps$. In this paper, we revisit this problem. Our main result is an $\eps$-edge differentially private algorithm for $k$-core decomposition which outputs the core numbers with no multiplicative error and $O(\text{log}(n)/\eps)$ additive error. This improves upon previous work by a factor of 2 in the multiplicative error, while giving near-optimal additive error. Our result relies on a novel generalized form of the sparse vector technique, which is especially well-suited for threshold-based graph algorithms; thus, we further strengthen the connection between distributed/parallel graph algorithms and differentially private algorithms.
    
    With a little additional work, we improve the additive error of the $\eps$-local edge differentially private $O(\log^2 n)$ round $(2+\eps)$-approximate$k$-core decomposition algorithm of~\cite{DLRSSY22} from $O\left(\frac{\log^3 n}{\eps}\right)$ to $O\left(\frac{\log n}{\eps}\right)$. We also present improved algorithms for densest subgraph and low out-degree ordering under differential privacy. For low out-degree ordering, we give an $\eps$-edge differentially private algorithm which outputs an implicit orientation such that the out-degree of each vertex is at most $d+O(\log{n}/{\eps})$, where $d$ is the degeneracy of the graph. This improves upon the best known guarantees for the problem by a factor of $4$ and gives near-optimal additive error. For densest subgraph, we give an $\eps$-edge differentially private algorithm outputting a subset of nodes that induces a subgraph of density at least ${D^*}/{2}-O(\text{log}(n)/\eps)$, where $D^*$ is the density for the optimal subgraph, improving over the additive error of all previous work.
\end{abstract}
\thispagestyle{empty}
\clearpage

\pagenumbering{arabic}
\section{Introduction}

The $k$-core decomposition is a fundamental graph characteristic that offers
deep insights into the structure of any graph. It is also an unique problem in 
the sense that it has managed to receive much more attention in recent years 
from the database and systems communities (for a comprehensive survey, see~\cite{malliaros2020core}) compared with the algorithms 
community (see \cite{DLRSSY22} and references therein). Such interests within the database and systems communities
largely stem from the many applications of solutions to the problem to various
practical tasks including community detection~\cite{Alvarez2005,Esfandiari2018,Ghaffari2019}, clustering and data mining~\cite{costa2011analyzing, shin2016corescope},
and various machine learning and graph analytic applications~\cite{Kabir2017, dhulipala2017julienne, dhulipala2018theoretically}. Due to the pace
by which practical implementations of $k$-core decomposition is currently being developed in many applications involving sensitive user data,
it is becoming ever more important to develop private algorithms for the problem.

The purpose of this paper is two-fold. First, we seek to fill in the gap for 
accurate, private and efficient algorithms for $k$-core decomposition.
Second, we seek to contribute to the growing but nascent body of theoretical 
work on formalizing the connection between low complexity parallel/distributed graph algorithms
and differentially private graph algorithms. The problem we study in this paper
have such widespread applications because the 
$k$-core decomposition is a very \emph{efficient} and \emph{accurate} measure 
of the set of important, well-connected vertices in a graph. Such a 
set of vertices is colloquially known as the ``core'' of the graph.
Formally, the $k$-core of a graph is the maximal induced
subgraph of $G$ where each vertex in the induced subgraph has degree 
at least $k$. Since this definition captures a single core, the $k$-core decomposition problem (sometimes also called the {\em coreness} problem) is to identify the {\em core number} of every vertex $v$, which is defined to be the largest $k$ for which $v$ is in the $k$-core, but not in the $(k+1)$-core.\footnote{We note that the $k$-core decomposition problem in the literature is sometimes used to refer to the problem of computing a hierarchy of $k$-cores based on graph connectivity~\cite{sariyuce2016fast}; note that it is not possible to emit such a hierarchy privately, so we (and prior work) focus on privately emitting core numbers for every vertex.} Identifying the coreness values of every vertex is an important way of classifying the structural importance of each node in the graph.

Private $k$-core decomposition was introduced in the recent work of 
Dhulipala \etal~\cite{DLRSSY22} with differential privacy guarantees;
they give an algorithm that
produces approximate core numbers for every
vertex with a multiplicative approximation
factor of $2+\eta$ and $O(\text{poly}(\log n)/\eps)$ additive
error, for any constant $\eta > 0$, in $O(\log n)$ rounds.
Differential privacy~\cite{DMNS06} is the gold standard 
for privacy in modern data analytics. 
They study the problem in the very strong 
\emph{local} edge model of differential privacy (LEDP). In particular, 
the local model guarantees the privacy of each individual adjacency list
such that no edge is known by any additional party beyond
the two vertices that are the endpoints of the edge. 
In other words, there is \emph{no notion of a
trusted third-party} in the local model of differential privacy. 
Since this model is inherently a distributed model, it is particularly important that the 
number of rounds of communication is small (e.g.\ $\poly(\log n)$). Another 
subtlety of the local model is that often the number of queries to a particular vertex
needs to be limited in order to achieve good approximation guarantees; this limits the set of 
algorithms we may consider for the problem.
By ensuring a small number of distributed rounds (and number of queries to any particular vertex), when
translating such algorithms to the central model (run on a 
centralized server on one machine), the running time
often becomes superlinear. In particular, the $k$-core decomposition
algorithm of~\cite{DLRSSY22}, when
directly translated to the central model, requires $\Omega(n^2 \log n)$
time. However, as there is no lower bound showing that a $(2+\eta)$-multiplicative approximation is necessary for privacy, it
was left open whether one can improve on this multiplicative
approximation factor. Additionally, their additive approximation
requires a fairly large poly-logarithmic factor and it was also left open whether this result
could be improved.

Although it is often the case that algorithms in the local model often have worse approximation
guarantees than algorithms in the central model for the same problem, it is not the case in this paper.
In this paper, we revisit the differentially private $k$-core decomposition problem 
and provide improved utility guarantees in both the \emph{central and local edge}
models of differential privacy. The central model assumes a trusted curator which has access to the private data and carefully processes it before releasing the privatized information. We work in the differential privacy models under edge-differential privacy. Our $k$-core decomposition algorithm in particular provides better utility guarantees than the recent work of~\cite{DLRSSY22}, providing an approximation with 
no multiplicative factor of the coreness values of every vertex with only an 
$O(\log n / \eps)$ additive error with high probability. Our algorithm runs in linear $O(n + m)$ time
but may use $\Omega(n)$ rounds in the distributed setting. However, we provide another algorithm in the distributed
setting that further improves the additive error of the 
$O(\log^2 n)$ round, $(2+\eta)$-approximation algorithm of~\cite{DLRSSY22} from $O\left(\frac{\log^3 n}{\eps}\right)$
to $O\left(\frac{\log n}{\eps}\right)$.
All algorithms presented in this paper are based on a novel, simple
generalized formulation of the AboveThreshold procedure 
from the sparse vector technique~\cite{dwork2009complexity} which we call the \emph{multidimensional
AboveThreshold technique (MAT).}
Like most existing algorithms for $k$-core, our algorithm is based on iterative peeling, i.e.\ iteratively removing the lowest degree vertices from the graph until all remaining vertices have large enough degree.
We show that our multidimensional version of the sparse vector technique can be easily adapted to such
threshold-based algorithms. Since these threshold-based algorithms are used frequently in distributed/parallel 
algorithms, our paper further strengthens the link between parallel/distributed graph algorithms and differentially
private algorithms via this general technique.

\subsection{Our Contributions}\label{sec:contributions}

Given an undirected graph $G=(V,E)$, let $n=|V|$ be the number of vertices and $m =|E|$ be the number of edges in the graph. We can define the $k$-core of a graph and core numbers of the vertices as follows:

\begin{definition}
A $k$-core of the graph is a maximal induced subgraph $H$ of $G$ where each vertex in $H$ has induced degree at least $k$. A vertex $v$ has {\em core number} (or {\em coreness}) $k$ if $v$ is in the $k$-core but not the $(k+1)$-core of $G$.
\end{definition}

When outputting approximate core numbers of a graph, one usually only requires a multiplicative approximation. In our setting, the algorithms need to additionally incur additive error in order to preserve differential privacy. Our notion of approximate $k$-core decompositions is defined as follows:

\begin{definition}
Given a vertex $v$, let $k(v)$ be the core number of $v$ and suppose the algorithm outputs $\hat{k}(v)$ as the approximate core number of $v$. Then we say that the algorithm gives a {\em $(\phi, \zeta)$-approximation} to the core numbers if we have $k(v) - \zeta \leq \hat{k}(v) \leq \phi\cdot k(v) + \zeta$ for all $v\in V$.
\end{definition}

Our main result is a private algorithm for $k$-core decomposition which improves both the multiplicative and additive approximation guarantees of the algorithm in \cite{DLRSSY22}. We adapt a standard peeling algorithm for the $k$-core decomposition using a clever generalization of the sparse vector technique to satisfy $\eps$-edge differential privacy and shows that the resulting algorithm achieves a $(1,O(\log{n} / \eps))$-approximation to the core numbers. Since the maximum $k$-core number is within a factor of $2$ of the densest subgraph in any graph by the Nash-Williams theorem, the additive lower bound given by~\cite{NV21, AHS21} also translates to a $\Omega(\sqrt{\log(n)/\eps})$ additive error lower bound for the $k$-core decomposition problem, so our result is nearly tight.

\begin{theorem}
There is an algorithm which gives $(1, O(\frac{\log(n)}{\eps}))$-approximate core numbers that is $\eps$-(local) edge differentially private and guarantees the approximation with high probability.
\end{theorem}

We also give another $\eps$-local edge differentially private algorithm that obtains the desired $\poly(\log n)$ round complexity.

\begin{theorem}
    There is an algorithm which gives $(2+\eta, O(\frac{\log n}{\eps}))$-approximate core numbers (for any constant 
    $\eta > 0$) that is $\eps$-local
    edge differentially private, runs in $O(\log^2 n)$ rounds, and guarantees the approximation with high probability.
\end{theorem}

In addition to $k$-core decomposition, we study the closely related densest subgraph problem, which we define and state our results for next. The density of a subgraph, and the densest subgraph problem are defined as follows:

\begin{definition}
Given an undirected graph $G(V,E)$, and a subgraph $S = (V_S, E_S)$ of $G$ where $E_S$ is the induced 
subgraph of $G$ on $V_S$, the {\em density} of $S$ is defined to be $d(S) = |E_S|/|V_S|$.
The {\em densest subgraph problem} is to find a subgraph $H_{\max}$ of $G$ which maximizes the density.
\end{definition}

As stated above, we cannot output the exact densest subgraph without violating privacy. Thus, we define an approximate densest subgraph, where $d(H)$ denotes the density of $H$.

\begin{definition}
Let the optimal density of a densest subgraph in $G$ be $D^{*}$. A subgraph $H$ is a $(\phi, \zeta)$-approximation to the densest subgraph if $d(H) \geq \frac{D^{*}}{\phi} - \zeta$ for $\phi$, $\zeta \geq 1$.
\end{definition}

Our main result for densest subgraph is given next, which follows from our $k$-core decomposition result.

\begin{theorem}
There is an algorithm which gives a $(2, O(\log n / \eps))$-approximation for the densest subgraph problem that is $\eps$-(local) edge differentially private and achieves the approximation guarantee with high probability.
\end{theorem}

We also give a differentially private algorithm for the \emph{low out-degree ordering} problem (\cref{def:low-outdeg}). The low out-degree ordering
problem seeks to output an ordering of the vertices in the graph to minimize the out-degree of any vertex when edges are oriented
from vertices earlier in the ordering to later in the ordering. It is well-known that the best ordering(s) gives out-degree upper bounded
by the degeneracy, $d$, of the graph.

\begin{definition}[$(\multfactor, \addfactor)$-Approximate Low Out-Degree Ordering]\label{def:low-outdeg}
    Let $\order = [v_1, v_2, \dots, v_n]$ be a total ordering of nodes in a
    graph $G = (V, E)$.
    The ordering $\order$ is an \defn{$(\multfactor, \addfactor)$-approximate
    low out-degree ordering} if
    orienting edges from earlier nodes to later nodes in $\order$ 
    produces out-degree at most
    $\multfactor \cdot \degen + \addfactor$.
\end{definition}

\begin{theorem}
    There is an $\eps$-(local) edge differentially private algorithm that gives a $(1, O(\log n/\eps))$-approximate
    low out-degree ordering of the vertices. 
\end{theorem}

As stated in our theorems, 
all of the algorithms in our paper can also be implemented in the local edge DP (LEDP) model with the same approximation and additive error guarantees as in the central model.\footnote{For the algorithms that may take $\Omega(n)$ rounds in the local model, we put ``(local)'' in the theorems, since these algorithms are more suitable for the central model.} Furthermore, all of our algorithms can be implemented (centrally) in near-linear time if we lose a multiplicative factor of $1+\eta$ for any given constant $\eta>0$. This nearly matches the best running times for 
sequential non-private $k$-core decomposition algorithms, while maintaining a near-optimal utility for private algorithms.

\subsection{Related Work}\label{sec:relwork}
The three problems our paper studies, $k$-core decomposition, densest subgraph and low out-degree ordering, 
are closely related and have been studied side by side in the literature. For instance, the classic (non-private) $2$-approximation to the densest subgraph problem by Charikar~\cite{Charikar00} works by performing an iterative greedy peeling process. Subsequent work improved the approximation factor by using the multiplicative weight update framework~\cite{BGM14,SuVu20}. There has also been a large body of work on dynamic densest subgraph, with recent breakthrough results of~\cite{sawlani2020near,chekuri20221,chekuri2024adaptive} 
yielding $O(\mathsf{poly}(\log n))$ update time while maintaining a $(1+\eta)$-approximate densest subgraph. 
The works of \cite{DLRSSY22,DKLV23} rely heavily on ideas from this literature. 

On the privacy side, the first results for private densest subgraph are the works of Nguyen and Vullikanti~\cite{NV21} and
Farhadi et al.~\cite{FHO21}. \cite{NV21} give algorithms which achieve a 
$O(2 + \eta, \log(n) \log(1/\delta)/{\eps})$-approximation
under $(\eps, \delta)$-differential privacy. On the other hand, ~\cite{FHO21} give algorithms which achieve a 
$(2+\eta,O({\log^{2.5} (n)\log(1/\sigma)}/{\eps}))$-approximation with probability $1-\sigma$ 
under $\eps$-differential privacy. More recently, \cite{DLRSSY22} gave a $(1+\eta, O(\mathsf{poly}(\log n)/\eps))$
approximation achieving $\eps$-differential privacy by adapting the muliplicative weight update distributed 
algorithms of~\cite{BGM14,SuVu20} to the 
private setting; however, the private algorithm is not distributed and runs in $O((n+m)\log^3 n)$ time.
Both~\cite{NV21} and~\cite{FHO21} provide
lower bounds on the additive error of any private densest subgraph algorithm; in particular, their lower bounds 
state that any $\eps$-DP densest subgraph algorithm achieving high-probability accuracy guarantees
must have additive error $\Omega(\sqrt{\log(n)/\eps})$.
Our densest subgraph result yields the same multiplicative factor as~\cite{NV21, FHO21} but improves
the additive factor; on the other hand, we incur a worse multiplicative factor compared with~\cite{DLRSSY22}
but a better additive factor.

For $k$-core decomposition, to the best of our knowledge \cite{DLRSSY22} is the first
work to study this problem; they give the first private algorithm achieving a $(2 + \eta, O(\mathsf{poly}(\log n)/\eps))$-approximation in the local edge-differential privacy setting by adapting a recently developed {\em parallel} 
locally-adjustable dynamic algorithm from~\cite{LSYDS22}.
Since their approach is based on a parallel algorithm, improving on this result, e.g., to achieve any $(2-\eta)$-approximation seems
challenging due to a known $\mathsf{P}$-completeness result for degeneracy~\cite{anderson84pcomplete} (which reduces to $k$-core decomposition).
We take a different approach in this paper, and adapt a variant of the classical peeling algorithm~\cite{jacm/MatulaB83}, which outputs exact core numbers with high probability. We note that this algorithm is {\em not parallelizable},
and indeed, the depth (i.e., longest chain of sequential dependencies) of this algorithm could be $\Omega(n)$ in the worst case.

\paragraph{Concurrent, Independent Work.} The recent concurrent, independent work of Dinitz et al.~\cite{DKLV23} provides a large set of 
novel results on $(\eps, \delta)$-DP, $\eps$-LEDP, and $(\eps, \delta)$-LEDP densest subgraph, improving on the approximation guarantees of 
previous results. Their work is based on a recent result of Chekuri, Quanrud, and Torres~\cite{CQT22} and the parallel densest subgraph
algorithm of Bahmani, Kumar, and Vassilvitskii~\cite{BKV12}. They offer improvements in the bounds in the central DP setting via the
private selection mechanism of Liu and Talwar~\cite{liu2019private}. Although~\cite{DKLV23} also studies the densest subgraph problem (closely
related to the $k$-core decomposition problem), their techniques are different from those presented in this paper; furthermore, any $c$-approximate
algorithm for densest subgraphs translates to a $2c$-approximate algorithm for finding the degeneracy of the input graph (which is the value
of the maximum $k$-core). Thus, no densest subgraph algorithm can obtain a better than $2$-approximation of the maximum $k$-core value.
\section{Differential Privacy Background}

We give some basic definitions and background on differential privacy in this section. We refer the interested reader to~\cite{dwork2014algorithmic} for proofs of these fundamental concepts. Since we will only be working with graphs, we will state these results in terms of edge differential privacy for graph algorithms; we will also use the terms differential privacy and edge differential privacy interchangeably. We start with the definition of edge-neighboring graphs and edge-differential privacy. Edge-neighboring graphs is a well-studied model for 
differential privacy (see e.g.~\cite{li2023private} for a survey of such results) and protects the privacy of connections with highly sensitive information like 
disease transmission graphs.

\begin{definition}
We say that two graphs $G_1=(V_1,E_1)$ and $G_2=(V_2,E_2)$ are \emph{edge-neighboring} if $V_1=V_2$ and $|E_1\oplus E_2|=1$ (i.e., they have the same vertex set and they differ by exactly one edge). We often use the notation $G_1\sim G_2$ to denote that $G_1,G_2$ are edge-neighboring.
\label{def:edge-nbhr}
\end{definition}

\begin{definition}
Let $\mathcal{G}$ denote the set of undirected graphs. We say an algorithm $\mathcal{M}:\mathcal{G}\to\mathcal{Y}$ is \emph{$(\eps,\delta)$-edge differentially private} if for all edge-neighboring graphs $G\sim G^\prime$ and every $S\subseteq\mathcal{Y}$, we have $$\Pr[\mathcal{M}(G)\in S]\le\exp(\eps)\cdot\Pr[\mathcal{M}(G^\prime)\in S]+\delta.$$
If $\delta=0$, we say $\mathcal{M}$ is $\eps$-edge differentially private.
\end{definition}

Next, we define two important properties of differential privacy: composition and post-processing. The composition property enables us to easily combine private subroutines into a larger algorithm which maintains some (weaker) privacy guarantees. The post-processing property guarantees that the output of a private algorithm cannot be post-processed to become less private without using information from the private data.

\begin{lemma}
Let $\mathcal{M}_1:\mathcal{G}\to\mathcal{Y}_1$ and $\mathcal{M}_2:\mathcal{G}\to\mathcal{Y}_2$ be $\eps_1$- and $\eps_2$-edge differentially private mechanisms, respectively. Then $\mathcal{M}:\mathcal{G}\to\mathcal{Y}_1\times\mathcal{Y}_2$ defined by $\mathcal{M}=(\mathcal{M}_1,\mathcal{M}_2)$ is $(\eps_1+\eps_2)$-edge differentially private.
\label{lem:basic-comp}
\end{lemma}

\begin{lemma}
Let $\mathcal{M}:\mathcal{G}\to\mathcal{Y}$ be an $\eps$-edge differentially private mechanism.
Let $f:\mathcal{Y}\to\mathcal{Z}$ be an arbitrary randomized mapping. Then $f\circ\mathcal{M}:\mathcal{G}\to\mathcal{Z}$ is still $\eps$-edge differentially private.\label{lem:post}
\end{lemma}

\subsection{Local Edge Differential Privacy}

The local model of differential privacy studies the setting where all private information are kept private by the individual parties, i.e.\ there exists no trusted curator. For graphs, the local model is called local edge differential privacy as defined in~\cite{DLRSSY22}. 
We now define this formally.

\begin{definition}\label{def:local-randomizer}
    An \defn{$\eps$-local randomizer} $R: \adj \rightarrow \rangeout$ for node $v$ is an $\eps$-edge differentially private 
    algorithm that takes as input the set of its neighbors $N(v)$, represented by
    an adjacency list $\adj = (b_1, \dots, b_{|N(v)|})$. 
\end{definition}

In other words, a local randomizer for node $v$ is an edge-private algorithm which is run on the neighborhood of $v$. The information released via local randomizers is public to all nodes and the curator. 
The curator performs some computation on the released information and makes the result public. The overall computation is formalized via the notion of the transcript. 

\begin{definition}\label{def:LEDP}
A \emph{transcript} $\pi$ is a vector consisting of 4-tuples $(S^t_U, S^t_R, S^t_\eps, S^t_Y)$ -- encoding the set of parties chosen, set of local randomizers assigned, set of randomizer privacy parameters, and set of randomized outputs produced -- for each round $t$. Let $S_\pi$ be the collection of all transcripts and $S_R$ be the collection of all randomizers. Let $\text{STOP}$ denote a special character indicating that the local randomizer's computation halts.
A \emph{protocol} is an algorithm $\alg: S_\pi \to 
(2^{[n]} \times 2^{S_R} \times 2^{\mathbb{R}^{\geq 0}} \times 2^{\mathbb{R}^{\geq 0}})\; \cup \{\text{STOP}\}$
mapping transcripts to sets of parties, randomizers, and randomizer privacy parameters. The length of the transcript, as indexed by $t$, is its round complexity. Given $\eps> 0$,  a randomized protocol $\alg$ on graph $G$ is \emph{$\eps$-local edge differentially private} if the algorithm that outputs the entire transcript generated by $\alg$ is $\eps$-edge differentially private on graph $G.$
\end{definition}

The definition is difficult to parse, but it naturally corresponds with the intuition of dealing with an untrusted curator (see also the discussion in~\cite{DKLV23}). At the beginning, the curator only knows the node set $V$ and each node knows its neighborhood. In each round, a subset of the nodes perform some computation using local randomizers on their local edge information, the outputs from previous rounds, and the public information. The nodes then output something which can be seen by all other nodes and the curator. This released information is differentially private since it is released by local randomizers. It is the curator's choice for whom to query, what local randomizer they use, and what privacy parameters the randomizer uses. Since the curator's choice only depends on the transcripts from the previous rounds, this is exactly the definition of a protocol in the above definition.
\section{Generalized Sparse Vector Technique}\label{sec:mat}

We introduce a simple, novel multidimensional generalization of the AboveThreshold mechanism, which is the basis of the algorithms in this paper. In the standard AboveThreshold mechanism from Section 3.6 in~\cite{dwork2014algorithmic}, we are given as input a threshold $T$ and an adaptive sequence of sensitivity $1$ queries $f_1,f_2,\ldots$, and the goal is to output the first query which exceeds the threshold. In the multidimensional version, we have a $d$-dimensional vector of thresholds $\vec{T}=(T_1,\ldots,T_d)$ and an adaptive sequence of sensitivity $\Delta$ vector-valued queries $\vec{f}_1,\vec{f}_2,\ldots$, and the goal is to output the first query for which the $j^{th}$ coordinate exceeds the threshold $T_j$, for each $j\in[d]$. Here $\Delta$ is defined as the sensitivity for \emph{all} vectors.
We give a natural generalization of the 1-dimensional to the multidimensional AboveThreshold mechanism in \Cref{alg:multidimensional Above Threshold} and give its privacy proof below.

\begin{algorithm}[t]
\caption{Multidimensional AboveThreshold (MAT)}
\label{alg:multidimensional Above Threshold}
\textbf{Input:} Private graph $G$, adaptive queries $\{\vec{f}_1, \dots, \vec{f}_n\}$, threshold vector $\vec{T}$, privacy $\eps$, 
sensitivity $\Delta$.\\
\textbf{Output:} A sequence of responses $\{\vec{a}_1, \dots, \vec{a}_n\}$ where $a_{i,j}$ indicates if $f_{i,j}(G)\ge \vec{T}_j$\\
\begin{algorithmic}[1]
\FOR{$j=1,\ldots,d$}
    \STATE $\hat{T}_j\leftarrow \vec{T}_j+\text{Lap}(2\Delta/\eps)$
\ENDFOR
\STATE
\FOR{each query $\vec{f}_i \in \{\vec{f}_1, \dots, \vec{f}_n\}$}
\FOR{$j=1,\ldots,d$}
\STATE Let $\nu_{i,j}\leftarrow\text{Lap}(4\Delta/\eps)$
\IF{$f_{i,j}(G)+\nu_{i,j}\ge \hat{T}_j$}
\STATE Output $a_{i,j}=\top$
\STATE Stop answering queries for coordinate $j$
\ELSE
\STATE Output $a_{i,j}=\bot$
\ENDIF
\ENDFOR
\ENDFOR
\end{algorithmic}
\end{algorithm}

\begin{theorem}
    \Cref{alg:multidimensional Above Threshold} is $\eps$-differentially private.\label{thm:privacy AboveThreshold}
\end{theorem}
\begin{proof}
Fix $G\sim G^\prime$ to be arbitrary edge-neighboring graphs. We first define some notation for simplicity in the remainder of the proof. Let $\mathcal{A}(G)$ and $\mathcal{A}(G')$ denote the random variable representing the output of \Cref{alg:multidimensional Above Threshold} when run on $G$ and $G'$, respectively. The output of the algorithm is a realization of the random variables, where for each coordinate $j\in [d]$, the output is $a_{\cdot j}\in\{\top,\bot\}^{r(j)}$ and has the form that for all $i<r(j)$, $a_{i,j}=\bot$ and $a_{r(j),j}=\top$. Our goal will be to show that
for any output $a$, we have
$$\Pr[\mathcal{A}(G)=a]\le \exp(\eps)\cdot\Pr[\mathcal{A}(G')=a].$$
Observe that there are two sources of randomness in the algorithm: the noisy thresholds $\hat{T}_j$ for each $j\in[d]$ and the perturbations of the queries $\nu_{i,j}$ for each $j\in[d]$, $i\in[r(j)]$. We will fix the random variables $\nu_{i,j}$ for each $j\in[d]$ and $i<r(j)$ by simply conditioning on their randomness; for neighboring graphs, we have a coupling between corresponding variables. 
For simplicity, we omit this from notation. The remaining random variables are $\hat{T}_j$ and $\nu_{r(j),j}$ for each $j\in [d]$; we will be taking probabilities over this randomness, but this will remain implicit in our notation. 
The main observation needed is that the output of $\mathcal{A}$ is uniquely determined by the iterations $r(j)$ of the first time we have $f_{i,j}(G)+\nu_{i,j}\ge \hat{T}_j$ for each $j\in [d]$. Thus, we can define the (deterministic, due to conditioning on the probabilities) quantity
$$g_j(G)\coloneqq \text{max}\{f_{i,j}(G)+\nu_{i,j}:i<r(j)\}$$
for each $j\in[d]$ so that the event that $\mathcal{A}(G)=a$ is equivalent to the event where 
\begin{align}
    f_{r(j),j}(G)+\nu_{r(j),j} \geq \hat{T}_j\text{ and } \hat{T}_j > g_j(G)\text{ for each } j\in[d].\label{eq:event}   
\end{align}
Let $p(\cdot)$ and $q(\cdot)$ denote the density functions of the random vectors consisting of $\hat{T}(j)$ and $\nu_{r(j),j}$ for each $j\in[d]$. Given the above observation, we can write the following, where $\hat{T}_j\in(f_{r(j),j}(G)+\nu_{r(j),j},g_j(G))$ indicates the aforementioned event,
\begin{align}
    \Pr[\mathcal{A}(G)=a]&=\Pr[\hat{T}_j\in(f_{r(j),j}(G)+\nu_{r(j),j},g_j(G))\text{ $\forall j\in [d]$}]\nonumber\\
    &=\int_{\mathbb{R}^d}\int_{\mathbb{R}^d}p(\vec{t})q(\vec{v})\cdot\mathbbm{1}\{t_j\in(f_{r(j),j}(G)+v_j,g_j(G))\text{ $\forall j\in [d]$}\}\,d\vec{v}\,d\vec{t}\label{eq:star}
\end{align}
Now, we make a change of variables. Define 
\begin{align*}
    \vec{v}^\prime_j&=\vec{v}_j+g_j(G)-g_j(G^\prime)+f_{r(j),j}(G')-f_{r(j),j}(G)\\
    \vec{t}^\prime_j&=\vec{t}_j+g_j(G)-g_j(G^\prime)
\end{align*}
for each $j\in [d]$. Since we have conditioned on the randomness of $\nu_{i,j}$ for each $j\in[d]$ and $i<r(j)$, we have $\|\vec{v}^\prime-\vec{v}\|_1\le 2\Delta$ and $\|\tilde{t}^\prime-\tilde{t}\|_1\le \Delta$ since the vectors $f_{i\cdot}(G)$ have sensitivity $\Delta$, implying that the vectors $g_{\cdot}(G)$ also have sensitivity $\Delta$. Also observe that the change of variables only changes the vectors $\vec{v}$ and $\vec{t}$ by a constant (since we are still conditioning on the randomness of $\nu_{i,j}$) so $d\vec{v}\,d\vec{t}=d\vec{v}'\,d\vec{t}'$. Thus, we can apply the change of variables to rewrite (\ref{eq:star}) as: 
\begin{align}
    \int_{\mathbb{R}^d}\int_{\mathbb{R}^d}p(\vec{t}^\prime)q(\vec{v}^\prime)\mathbbm{1}\{{t}^\prime_j\in(f_{r(j),j}(G)+{v}^\prime_j,g_j(G))\text{ $\forall j\in [d]$}\}\,d\vec{v}\,d\vec{t}.\nonumber
\end{align}
Plugging in the definitions of $t'_j$ and $v'_j$ and simplifying, this is equivalent to
\begin{align}
    \int_{\mathbb{R}^d}\int_{\mathbb{R}^d}p(\vec{t}^\prime)q(\vec{v}^\prime)\mathbbm{1}\{{t}_j\in(f_{r(j),j}(G^\prime)+{v}_j,g_j(G^\prime))\text{ $\forall j\in [d]$}\}\,d\vec{v}\,d\vec{t}.\nonumber
\end{align}
Now, recall that we showed $\|\vec{v}^\prime-\vec{v}\|_1\le 2\Delta$ and $\|\vec{t}^\prime-\vec{t}\|_1\le \Delta$. This implies that $p(\vec{t}^\prime)\le\exp(\eps/2)\cdot p(\vec{t})$ and $q(\vec{v}^\prime)\le\exp(\eps/2)\cdot q(\vec{v})$ by the properties of the Laplace distribution. We can thus upper bound the above as
\begin{align}
    \exp(\eps)\cdot\int_{\mathbb{R}^d}\int_{\mathbb{R}^d}p(\vec{t})q(\vec{v})\mathbbm{1}\{t_j\in(f_{r(j),j}(G^\prime)+{v}_j,g_j(G^\prime))\text{ $\forall j\in [d]$}\}\,d\vec{v}\,d\vec{t}\nonumber
\end{align}
By the definition of probability, the above is exactly equal to
\begin{align*}
    \exp(\eps)\cdot\textstyle\Pr[\hat{T}_j\in(f_{r(j),j}(G^\prime)+\nu_{r(j),j},g_j(G^\prime))\text{ $\forall j\in [d]$}].
\end{align*}
Finally, by our observation in (\ref{eq:event}), this is equal to
\begin{align*}
    \exp(\eps)\cdot\Pr[\mathcal{A}(G')=a].
\end{align*}
Combining the sequence of equalities and inequalities, we have shown that
$$\Pr[\mathcal{A}(G)=a]\le \exp(\eps)\cdot\Pr[\mathcal{A}(G')=a].$$
Since $G,G^\prime$ were arbitrary edge-neighboring graphs, this completes the proof.     
\end{proof}

\subsection{Local Privacy Guarantees}

In our algorithms, we apply the multidimensional AboveThreshold mechanism in the special case where the queries $\vec{f}_i$ have $d=n$ coordinates, and each coordinate of the query corresponds to a node $u$ in the graph $G$ and only depends on the edges $e=(u,v)$ for $v\in V-\{u\}$. In other words, we have an instance of the AboveThreshold on each node $u\in V$, where the data used for the AboveThreshold instance is only the local (edge) data of $u$\footnote{Note that in this setting, an intuitive way to see the privacy guarantee is to view the multidimensional mechanism as a parallel composition of $n$ standard AboveThreshold mechanisms. Since each edge $(u,v)$ is only used in two of the AboveThreshold mechanisms, this implies $2\eps$-differential privacy by parallel composition. However, this is not immediately correct since the queries $\vec{f}_i$ are adaptive on the output of all $n$ coordinates of the previous queries, which is not allowed in standard parallel composition. This would require a concurrent version of parallel composition, which does not yet exist.}. We will show that in this setting, the multidimensional AboveThreshold mechanism can be implemented locally to satisfy local edge-differential privacy. 

For clarity of notation, we index the coordinates of the queries and threshold vector by nodes $v\in V$ instead of indices $j\in[n]$. That is, each query consists of $\vec{f}_{i,v}$ for $v\in V$ and the threshold vector consists of $\vec{T}_{v}$ for $v\in V$. We will now present the changes needed for the local implementation. In Lines 1--3, let each node $u\in V$ compute and store its noisy threshold $\hat{T}_u$ using the public threshold $\vec{T}_u$. Then for each query $\vec{f}_i$ in lines 5--15, each node $u\in V$ can sample its own noise $\nu_{i,u}$ and check the condition $f_{i,u}(G)+\nu_{i,u}\ge\hat{T}_u$. The pseudocode is given in \Cref{alg:local multidimensional Above Threshold}, and we now show that this is locally edge-differentially private.

\begin{algorithm}[t]
\caption{Local Multidimensional AboveThreshold}
\label{alg:local multidimensional Above Threshold}
\textbf{Input:} Private graph $G$, adaptive queries $\{\vec{f}_i\}$, threshold vector $\vec{T}$, privacy $\eps$, 
sensitivity $\Delta$.\\
\textbf{Output:} A sequence of responses $\{\vec{a}_i\}$ where $a_{i,j}$ indicates if $f_{i,j}(G)\ge \vec{T}_j$\\
\begin{algorithmic}[1]
\FOR{$v\in V$}
    \STATE Node $v$ computes $\hat{T}_v\leftarrow \vec{T}_v+\text{Lap}(2\Delta/\eps)$ and stores it
\ENDFOR
\STATE
\FOR{each query $\vec{f}_i$}
\FOR{$v\in V$}
\STATE Node $v$ samples $\nu_{i,v}\leftarrow\text{Lap}(4\Delta/\eps)$
\IF{$f_{i,v}(G)+\nu_{i,v}\ge \hat{T}_v$}
\STATE Node $v$ outputs $a_{i,j}=\top$
\STATE Node $v$ outputs $\text{STOP}$, and stops answering queries
\ELSE
\STATE Node $v$ outputs $a_{i,v}=\bot$
\ENDIF
\ENDFOR
\ENDFOR
\end{algorithmic}
\end{algorithm}

\begin{theorem}\label{thm:mat-ledp}
    \Cref{alg:local multidimensional Above Threshold} is $\eps$-locally edge differentially private.
\end{theorem}
\begin{proof}
    First, we need to show that the local randomizers are in fact edge-differentially private. This can be seen because each randomizers' output is a subset of the output of \Cref{alg:multidimensional Above Threshold}. In other words, the output is a post-processing of \Cref{alg:multidimensional Above Threshold}, so the privacy guarantees follow from \Cref{thm:privacy AboveThreshold} and \Cref{lem:post}. Next, we argue that the transcript is also $\eps$-edge differentially private. Recall that the transcript consists of the set of parties chosen, set of local randomizers assigned, set of privacy parameters assigned, and the set of outputs. In the algorithm, the set of parties chosen, set of local randomizers assigned, and set of privacy parameters assigned are all functions of the outputs from the previous rounds and the public information. Thus, it suffices to prove that the outputs are differentially private by post-processing (\Cref{lem:post}). But this was already proven in \Cref{thm:privacy AboveThreshold}, so we are done.
\end{proof}
\section{Differentially Private $k$-Core Decomposition}

In this section, we give our improved algorithm for differentially private $k$-core decomposition. We first present a variant of the classical (non-private) algorithm for $k$-core decomposition in \Cref{sec:optimal-variant}, and then show how to make it private in \Cref{sec:optimal-private}. Finally, we show how to implement it in near-linear time in \Cref{sec:optimal-efficient}, while incurring only $(1+\eta)$-multiplicative error.

\subsection{A Variation of the Classical Algorithm}
\label{sec:optimal-variant}

The classical peeling algorithm begins with the full vertex set $V$, which is the $0$-core of the graph. Given the $(k-1)$-core, the algorithm computes the $k$-core via an iterative peeling process: the algorithm repeatedly removes all nodes $v$ for which the induced degree is less than $k$, and labels the nodes which remain as being part of the $k$-core. Running this for $k$ from 1 up to $n$ gives the full algorithm, the pseudocode of which is given in \Cref{alg:classical-algorithm}.

\begin{algorithm}[h]
\caption{Threshold-Based $k$-core Decomposition Algorithm}
\label{alg:classical-algorithm}
\textbf{Input:} Graph $G=(V,E)$.\\
\textbf{Output:} $k$-core value of each node $v\in V$\\
\begin{algorithmic}[1]
\STATE Initialize $V_0\leftarrow V$, $t\leftarrow 0$, $\hat{k}(v)\leftarrow 0$ for all $v\in V$ 
\FOR{$k=1,\ldots,n$}\label{line:for-k}
\REPEAT
\STATE $t\leftarrow t+1$, $V_t\leftarrow V_{t-1}$
\FOR{$v\in V_{t-1}$}
\IF{$d_{V_{t-1}}(v)< k$}\label{line:classical if condition}
\STATE $V_t\leftarrow V_t-\{v\}$\label{line:classical peeled}
\ENDIF{}
\ENDFOR{}
\UNTIL{$V_{t-1}-V_{t} = \emptyset$}
\STATE
\STATE Update the core numbers $\hat{k}(v)\leftarrow k$ for all nodes $v\in V_t$
\ENDFOR
\end{algorithmic}
\end{algorithm}

\begin{theorem}
    For each $v\in V$, the output $\hat{k}(v)$ given by \Cref{alg:classical-algorithm} is the core number of $v$.
\end{theorem}
\begin{proof}
    We will inductively show that the algorithm recovers the $k$-core of the graph. The base case of $k=0$ is easy. Now, assume that the algorithm finds the true $(k-1)$-core. Let $V(k)$ denote the subset of nodes which aren't removed in the iterative process for $k - 1$ in Line~\ref{line:for-k}. We have that each node $v\in V(k)$ has induced degree at least $k$ in $V(k)$, or else it would have been removed. Thus, we know that the core numbers $k(v)$ is at least $k$ for each $v\in V(k)$, so we have that $V(k)$ is a subset of the true $k$-core. Now, let $K$ denote the true $k$-core. Since the $k$-core is always a subset of the $(k-1)$-core by definition, each node $v\in K$ is in $V_t$ at the beginning of the iterations. Furthermore, we know that $v\in K$ is never removed from $V_t$ since the induced degree is always at least $k$ since $K\subseteq V_t$. Thus, we have that $V(k)$ is a superset of the true $k$-core as well. Thus, $V(k)$ is the true $k$-core for each $k$, so all nodes are labelled correctly.  
\end{proof}

\subsection{Private Implementation of the Algorithm}
\label{sec:optimal-private}

It is difficult to turn the classical algorithm into a differentially private one because it has $\Omega(n)$ iterations, which causes us to incur $\tilde{\Omega}(n)$ additive error when using basic composition (\Cref{lem:basic-comp}). In fact, this was cited  in~\cite{DLRSSY22} as the motivation for basing their algorithms on parallel/distributed algorithms for $k$-core decomposition, since those algorithms often have $\text{poly}\log({n})$ round-complexity. Our main observation is that the private version of \Cref{alg:classical-algorithm}     can be analyzed as a special case of the Multidimensional AboveThreshold mechanism, so it doesn't need to incur the $\Omega(n)$ additive error due to composition. 

\begin{algorithm}[h]
\caption{$\eps$-Differentially Private $k$-Core Decomposition}
\label{alg:optimal-algorithm}
\textbf{Input:} Graph $G=(V,E)$, privacy parameter $\eps>0$.\\
\textbf{Output:} An $(1,120\log(n)/\eps)$-approximate $k$-core value of each node $v\in V$\\
\begin{algorithmic}[1]
\STATE $V_0\leftarrow V$, $t\leftarrow 0$, $k=60\log{n}/\eps$.
\STATE Initialize $\hat{k}(v)\leftarrow 0$ for all $v\in V$ \label{line: initial labelling}
\FOR{$v\in V$}
\STATE $\tilde{\ell}(v)\leftarrow \text{Lap}(4/\eps)$
\ENDFOR{}
\STATE
\WHILE{$k\le n$}\label{line:optimal k-iteration}
\REPEAT
\STATE $t\leftarrow t+1$, $V_t\leftarrow V_{t-1}$
\FOR{$v\in V_{t-1}$}
\STATE $\nu_{t,v}\leftarrow\text{Lap}(8/\eps)$
\IF{$d_{V_{t-1}}(v)+\nu_{t,v}\le k+\tilde{\ell}(v)$}\label{line:optimal if condition}
\STATE $V_t\leftarrow V_t-\{v\}$\label{line:optimal peeled}
\ENDIF{}
\ENDFOR{}
\UNTIL{$V_{t-1}-V_{t} = \emptyset$}
\STATE
\STATE Update the core numbers $\hat{k}(v)\leftarrow k$ for all nodes $v\in V_t$
\STATE $k\leftarrow k+60\log{n}/\eps$
\ENDWHILE{}
\end{algorithmic}
\end{algorithm}

\begin{theorem}
    \Cref{alg:optimal-algorithm} is $\eps$-edge differentially private.\label{thm:optimal-privacy}
\end{theorem}
\begin{proof}
    We will show that \Cref{alg:optimal-algorithm} is an instance of the multidimensional AboveThreshold mechanism, implying that it is $\eps$-edge differentially private. Specifically, we will show that its output can be obtained by post-processing the output of an instance of \Cref{alg:multidimensional Above Threshold}. Indeed, consider the instance of the multidimensional AboveThreshold mechanism where the input graph is $G$, the privacy parameter is $\eps$, and the threshold vector $\vec{T}=\vec{0}$ is the zero vector. We will now inductively (and adaptively) define the queries. 
    
    For each iteration $t$, the $t^{th}$ query consists of the vector of $k-d_{V_{t-1}}(v)$ for each $v\in V$, where $V_0=V$ as in the algorithm; note that this matches with the queries on Line \ref{line:optimal if condition} of the algorithm. First, observe that the sensitivity of the vector containing $k-d_{V_{t-1}}(v)$ for each $v\in V$ is $\Delta=2$ for each iteration $t$, as needed in the algorithm. Next, observe that we can construct $V_t$ from $V_{t-1}$ using only the outputs $a_t(v)$ of the queries at the current iteration: given $v\in V_{t-1}$, we include $v$ in $V_t$ as well if and only if $a_t(v)=\bot$, which is equivalent to the condition in Line \ref{line:optimal if condition} of $d_{V_{t-1}}(v)+\nu_{t,v}\le k+\tilde{\ell}(v)$\footnote{Here, we implicitly assume that the randomness used by the algorithm and the AboveThreshold mechanism are the same. This is justified by a coupling of the random variables in the two algorithms. Specifically, we couple the noise added in Line 4 of the algorithm with the noise added in Line 2 of the AboveThreshold mechanism and couple the noise $\nu_{t,v}$ with $\nu_{i,j}$ in the AboveThreshold mechanism. It is easy to see that for $\Delta=2$, the random variables are exactly the same so the privacy guarantees of multidimensional AboveThreshold translates to privacy guarantees for our algorithm.}. Thus, this is a feasible sequence of adaptive queries for the AboveThreshold mechanism.
    
    Since the sequence of subsets $\{V_t\}$ are obtained by post-processing the outputs $a_t(v)$ of the AboveThreshold mechanism, they are $\eps$-differentially private by \Cref{thm:privacy AboveThreshold} and \Cref{lem:post}. By applying post-processing again, the sequence of pairs $(V_t,k_t)$ for each iteration $t$ is also $\eps$-differentially private since the sequence $k_t$ is public. Given the pairs $(V_t,k_t)$ for each iteration $t$, we can now recover the approximate core numbers which the algorithm outputs by setting the core numbers as $\hat{k}(v)=k_t$ for each node in $V_t-V_{t-1}$. It is easily verified that this gives the exact same output as \Cref{alg:optimal-algorithm}, so we have $\eps$-differential privacy for the algorithm by applying post-processing (\Cref{lem:post}) once again.
\end{proof}

\begin{theorem}
    \Cref{alg:optimal-algorithm} outputs $(1,\frac{120\log{n}}{\eps})$-approximate core numbers $\hat{k}(v)$ with probability $1-O(\frac{1}{n^2})$.\label{thm:optimal-utility}
\end{theorem}
\begin{proof}
By the density function of the Laplace distribution, we know that we have $|\nu_{t,u}|\le \frac{40\log{n}}{\eps}$ and $|\tilde{\ell}(u)|\le \frac{20\log{n}}{\eps}$ each with probability at least $1-\frac{1}{n^5}$. Since there are at most $O(n^3)$ such random variables, taking a union bound over all nodes $u\in V$ and all iterations of the loops, we have the above guarantee with probability at least $1-{O}\left(\frac{1}{n^2}\right)$. We condition on the event that the above inequalities hold true for each $t\in[T]$ and each $u\in V$ for the remainder of the proof.

Fix an arbitrary iteration $k$ of the while loop starting on Line \ref{line:optimal k-iteration}. 
Let $H$ be the set of nodes remaining in $V_t$ at the end of the while loop in Lines 8---16. We claim that ($i$) all nodes $u\in H$ have core number at least $k-\frac{60\log{n}}{\eps}$ and ($ii$) all nodes $u\not\in H$ have core number at most $k+\frac{60\log{n}}{\eps}$. To see ($i$), consider the subgraph $H$; we claim that each node $u\in H$ has induced degree at least $k-\frac{60\log{n}}{\eps}$ in $H$. Indeed, since each node in $H$ was not removed in the final iteration of the while loop in Lines 8---16, we have that $d_{H_t}(u)+\nu_{t,u}\ge k+\tilde{\ell}(u)$ for each $u\in H$. But since we have assumed $|\nu_{t,u}|\le\frac{40\log{n}}{\eps}$ and $|\tilde{\ell}(u)|\le\frac{20\log{n}}{\eps}$, our desired bounds follow directly by the triangle inequality. To see ($ii$), let's suppose for contradiction that the $k$-core value of $u$ is $k(u)>\ell+\frac{60\log{n}}{\eps^\prime}$. Then there exists a subgraph $u\in K\subseteq V$ where the induced degree of each node $v\in K$ is $d_K(v)\ge k(u)$. But for such a subgraph $K$, the condition in Line 12 will always be true (again, because $|\nu_{t,u}|\le\frac{40\log{n}}{\eps}$ and $|\tilde{\ell}(u)|\le\frac{20\log{n}}{\eps}$) so $K\subseteq H$. But since $u\in K$, this contradicts the fact that $u\not\in H$.

Using the above bounds, we can now prove our desired results. First, consider an arbitrary node $u\in V$ labeled $\hat{k}(u)$ within the while loop in Lines \ref{line:optimal k-iteration}---20 and not relabeled later. Since $u$ was labeled $\hat{k}(u)$ in the current iteration, we have that $k(u)\ge \hat{k}(u)-\frac{60\log{n}}{\eps^\prime}$. Since $u$ was not relabelled in the next iteration where the threshold was $k=\hat{k}(u)+\frac{60\log{n}}{\eps}$, the node $u$ was removed from $V_t$ at that iteration. Consequently, we have $k(u)\le k+\frac{60\log{n}}{\eps}=\hat{k}(u)+\frac{120\log{n}}{\eps}$ by what we just proved above. Since the output of our algorithm is $\hat{k}(u)$, the desired bounds in the theorem statement follow directly. Now, let's consider an arbitrary node $u\in V$ labeled $0$ at Line \ref{line: initial labelling} at the beginning of the algorithm and not relabeled later. Since the node $u$ was not relabelled at the first iteration where $k=\frac{60\log{n}}{\eps}$, it was removed from $V_t$ at that iteration, implying that $k(u)\le \frac{120\log{n}}{\eps}$. Hence, we have the desired approximation guarantees for all nodes.
\end{proof}

The above proof can be modified with appropriate constants to give the approximation with probability 
at least $1 - \frac{1}{n^c}$ for any $c \geq 1$, thus yielding our desired with high probability result.

\subsection{Efficient Implementation of the Algorithm}
\label{sec:optimal-efficient}

Though the current algorithm has strong utility guarantees, the running time is at least quadratic in $n$. We will make a slight modification to the algorithm, which will incur an additional multiplicative error of $1+\eta$ for an user chosen constant $\eta>0$, and show that a clever sampling trick can be used to implement the resulting algorithm in time nearly linear in the number of edges $m$. Via a similar proof as before, we can obtain an analogous set of privacy and utility guarantees. 

\begin{theorem}\label{thm:efficient-1}
    \Cref{alg:efficient-algorithm} has the following guarantees:
    \begin{itemize}
        \item It is $\eps$-edge differentially private.
        \item It outputs $(1+\eta,O(\log{n}/\eps))$-approximate core numbers with probability $1-O(\frac{1}{n^2})$.
    \end{itemize} 
\end{theorem}

\begin{algorithm}[h]
\caption{}
\label{alg:efficient-algorithm}
\textbf{Input:} Graph $G=(V,E)$, multiplicative error $\eta>0$, privacy parameter $\eps>0$.\\
\textbf{Output:} An $(1+\eta,60\log(n)/\eps)$-approximate $k$-core value of each node $v\in V$\\
\begin{algorithmic}[1]
\STATE $V_0\leftarrow V$, $t\leftarrow 0$, $k=60\log{n}/\eps$.
\STATE Initialize $\hat{k}(v)\leftarrow 0$ for all $v\in V$ \label{line: efficient initial labelling}
\FOR{$v\in V$}
\STATE $\tilde{\ell}(v)\leftarrow \text{Lap}(4/\eps)$
\ENDFOR{}
\STATE
\WHILE{$k\le n$}\label{line:efficient k-iteration}
\REPEAT
\STATE $t\leftarrow t+1$, $V_t\leftarrow V_{t-1}$
\FOR{$v\in V_{t-1}$}
\STATE $\nu_{t,v}\leftarrow\text{Lap}(8/\eps)$
\IF{$d_{V_{t-1}}(v)+\nu_{t,v}\le k+\tilde{\ell}(v)$}\label{line:efficient if condition}
\STATE $V_t\leftarrow V_t-\{v\}$\label{line:efficient peeled}
\ENDIF{}
\ENDFOR{}
\UNTIL{$V_{t-1}-V_{t} = \emptyset$}
\STATE
\STATE Update the core numbers $\hat{k}(v)\leftarrow k$ for all nodes $v\in V_t$
\STATE $k\leftarrow (1+\eta)\cdot k$
\ENDWHILE{}
\end{algorithmic}
\end{algorithm}

Now, we discuss the efficient sampling procedure. We claim that the naive implementation of the algorithm given in \cref{alg:optimal-algorithm} requires $\Omega(n^2)$ runtime in the worst case. Let us focus on the do-while loop in Lines 8--16 of \cref{alg:optimal-algorithm}. Clearly, the for-loop in Lines 10--15 iterates $\Theta(n)$ times for each iteration of the do-while loop. We claim that in the worst case, the do-while loop also requires $\Omega(n)$ iterations. Let us illustrate in the easy case where there is no noise (or equivalently, when we take $\eps$ to infinity). Consider the path graph, and observe that the degree of each node is $2$ except the endpoints which have degree $1$. If the threshold $\ell$ is 2, then only the endpoints are removed and the resulting graph will again be a path graph, now of size $n-2$. Clearly, it will take $n/2$ iterations before the do-while loop terminates. This example can be generalized to the noisy version by increasing the separation of the degrees from $1$ up to $4\log{n}$.

We will give a more efficient way to implement the sampling so that the algorithm can run in $\tilde{O}(m+n)$ time. The idea is as follows: given a vertex $v$ and a fixed threshold $k+\tilde{\ell}$, we can calculate the probability $p$ that the vertex is removed at a given iteration of the do-while loop. In fact, this probability remains the same for vertex $v$ unless a neighbor of $v$ is removed from $H_t$. Thus, whenever a neighbor of $v$ is removed, we can sample a (geometric) random variable for when the vertex $v$ is peeled assuming no more neighbors of $v$ are removed. Whenever another neighbor of $v$ is removed, we have to resample this geometric random variable, but the number of times this resampling must be done over all vertices $v$ scales linearly with the number of edges $m$ rather than quadratically with the number of nodes $n$. This is detailed below. In~\cref{alg:algorithm-efficient}, $\text{Geom}(q)$ is the \emph{geometric} distribution with parameter $q$.

\begin{algorithm}[h]
\caption{An efficient implementation of Lines 8--16 in \cref{alg:optimal-algorithm}}
\label{alg:algorithm-efficient}
\hspace*{\algorithmicindent} \textbf{Input:} Graph $G=(V,E)$, parameter $\eps>0$, threshold $k+\tilde{\ell}(v)$ for each $v\in V$.\\
\hspace*{\algorithmicindent} \textbf{Output:} $V_t$ with same distribution as Lines 8--16 of \cref{alg:optimal-algorithm}\\
\begin{algorithmic}[1]
\STATE $t\leftarrow0$, $V_0\leftarrow V$
\STATE $\text{to-remove}(t)\leftarrow\emptyset$ for each $t\in [n]\hfill$\textcolor{Green}{//set of nodes to remove from $V_{t-1}$ at time $t$}
\STATE $\text{remove-time}(v)\leftarrow n$ for each $v\in V$
\STATE $\text{updated}\leftarrow V$
\REPEAT
\FOR{$v\in\text{updated}$}
\STATE remove node $v$ from to-remove(remove-time$(v))$\hfill\textcolor{Green}{//updates the sampled time for removing $v$}
\STATE $q\leftarrow \Pr[d_{V_t}(v)+\text{Lap}(8/\eps)\ge k+\tilde{\ell}(v)]$
\STATE $\text{remove-time}(v)\leftarrow t+\text{Geom}(q)$
\STATE add node $v$ to to-remove(remove-time$(v))$
\ENDFOR
\STATE
\STATE $t\leftarrow t+1$, $V_t\leftarrow V_{t-1}$\hfill\textcolor{Green}{//this is only notational; we only store the set $V_t$ for a single $t$}
\STATE
\FOR{$v\in\text{to-remove}(t)$}\label{line:to-remove}
\STATE remove $v$ from $V_t$
\STATE add all neighbors of $v$ in $V_t$ to \text{updated}\hfill\textcolor{Green}{//need to update remove-time if degree changes}
\ENDFOR{}
\UNTIL{$V_{t-1}-V_{t} = \emptyset$}
\end{algorithmic}
\end{algorithm}

\begin{lemma}
    \cref{alg:algorithm-efficient} satisfies the following:
    \begin{itemize}  
        \item It gives the same distribution of $V_t$ as Lines 8--16 in \cref{alg:optimal-algorithm}.
        \item It can be implemented in $\tilde{O}(n+m)$ time.  
    \end{itemize}\label{lem:run time}
\end{lemma}
\begin{proof}
To see that the distribution is the same, we use induction on $t$. We claim that at each iteration $t$, the distribution of $V_t$ is the same for the two implementations of the algorithm. The statement is clear for $t=0$ since $V_0=V$ in both implementations. Now, assume that the distribution of $V_{t-1}$ is the same for both implementations; we will show that the distribution of $V_{t}$ conditioned on $V_{t-1}$ is also the same for both implementations. This will imply our desired result by induction.

Given $V_{t-1}$, let us compute that the probability that a node $v\in V_{t-1}$ is not included in $V_t$ in both implementations. Note that the event that a node is not included in $V_t$ is independent (conditioned on $V_t$) so showing that the probabilities are the same is sufficient to show the equivalence of (conditional) distributions. In the implementation in \cref{alg:optimal-algorithm}, the probability is clearly just
\begin{align}
    \Pr[d_{V_t}(v)+\text{Lap}(8/\eps)\ge k+\tilde{\ell}(v)].\label{eq:prob-remove}
\end{align}
Now, let's consider the implementation in \cref{alg:algorithm-efficient}. To understand when the node $v$ will be removed, we look at the final time $t^\prime$ that node $v$ is added to the set \textit{updated}. The node $v$ will be removed at time $t$ if the \textit{remove-time(v)} sampled by the algorithm is exactly $t$. The probability that this occurs conditioned on the event that $v$ was not removed at $V_{t-1}$ can be calculated as:
$$\Pr[t^\prime+\text{Geom}(q)=t \mid t^\prime+\text{Geom}(q)\ge t-1]=\Pr[\text{Geom}(q)=1]=q,$$
where the first equality holds by the memoryless property of the geometric distribution. Since no neighbors of $v$ were removed between iterations $t^\prime$ and $t$, the induced degrees $d_{V_{t^\prime}}(v)$ and $d_{V_t}(v)$ are the same. It is thus evident that the probability $q$ is the same as the expression in \cref{eq:prob-remove}, proving the desired claim.

Now, we will prove that the implementation given in \cref{alg:algorithm-efficient} runs in near-linear time. First, we note that all sets are implemented as binary search trees so adding and removing elements requires only $O(\log{n})$ time. To analyze the time complexity, observe that the total run time of Lines 6--11 (in an amortized sense) scales with the total number of nodes added to the set \textit{updated} (up to logarithmic factors). Since a node is only added to the set \textit{updated} when one of its neighbors is removed, the total number of nodes added to \textit{updated} during the entire algorithm scales with the number of edges. The time complexity of Lines 16 is $\tilde{O}(n)$ since each node is only removed once and the time complexity of Line 17 is $\tilde{O}(m)$ for the same reason as before. That completes the proof.
\end{proof}

\section{An Improved Analysis of \cite{DLRSSY22}}\label{sec:improved-distributed}

We use our MAT technique above to improve the algorithm of~\cite{DLRSSY22} so that we obtain an $\eps$-LEDP algorithm for 
\kc decomposition with additive error $O\left(\frac{\log n}{\eps}\right)$ while maintaining the original
$O(\log^2 n)$ round complexity.  We reframe Algorithm 3 of~\cite{DLRSSY22} 
within our MAT framework in the pseudocode given in~\cref{alg:improved-ledp-kcore}. We highlight the changes we make to the 
original algorithm in~\cite{DLRSSY22} in blue. 
Using the improved analysis for the 
MAT framework, we can obtain our desired error bounds.

\begin{algorithm}[htb!]
\caption{Improved $\eps$-LEDP $k$-Core Decomposition in $O(\log^2 n)$ Rounds Based on~\cite{DLRSSY22}}
\label{alg:improved-ledp-kcore}
\textbf{Input:} Graph $G=(V,E)$, privacy parameter $\eps>0$.\\
\textbf{Output:} $(2,c'\log(n)/\eps)$-approximate core numbers of each node in $G$ for sufficiently large constant $c' > 0$\\
\begin{algorithmic}[1]
\STATE Set {\color{blue} $\psi = 0.1\coren$ and $\lambda = \frac{2(30- \eta)\eta}{(\eta + 10)^2}$}.\\
{\color{blue} \FOR{$i = 1$ to $n$}
    \STATE $\tilde{\ell}(i) \leftarrow  \text{Lap}(4/\eps)$.\label{lkcsr-line:noise-1}
\ENDFOR}
\FOR{$\lcur = 0$ to $\numlevels - 1$}\label{lkcsr-line:iterate}
    \STATE $L_{\lcur + 1}[i] \leftarrow L_{\lcur}[i]$.\label{lkcsr-line:old-level}\\
    \IF{$L_\lcur[i] = \lcur$}\label{lkcsr-line:level-cur-level}
        \STATE Let $\nup_{i}$ be the number of neighbors $j \in \adj_i$ where $L_\lcur[j] =
                \lcur$.\label{lkcsr-line:compute-up}\\
        \STATE Sample {\color{blue} $X \leftarrow  \text{Lap}(8/\eps)$}.\label{lkcsr-line:sample-noise}\\
        \STATE Compute $\hnup_{i} \leftarrow \nup_{i} + X$.\label{lkcsr-line:compute-noisy-up}\\
        \IF{{\color{blue}$\hnup_{i} \leq \upexp^{\lfloor r/(\numgrouplevels)\rfloor} + \tilde{\ell}(i)$}}\label{lkcsr-line:move-up-condition}
            \STATE {\color{blue} Stop and set $A_i \leftarrow 0$.}\label{lkcsr-line:group-update-1}
        \ELSE
            \STATE {\color{blue} Continue and set $A_i \leftarrow 1$.}\label{lkcsr-line:group-update}\\
        \ENDIF
    \ENDIF
    \STATE $i$ \release $A_i$.\label{lkcsr-line:release-ai}\\
    \STATE $L_{\lcur + 1}[i] \leftarrow L_{\lcur}[i] + 1$ for every $i$ where $A_i = 1$.\label{lkcsr-line:update-list} 
\ENDFOR
\STATE Curator publishes $L_{\lcur + 1}$.\label{lkcsr-line:levels}\\ 
\STATE $C \leftarrow \emptyset$.\label{lkcsr-line:c-empty}\\
\FOR{$i = 1$ to $n$}
    \STATE Let $\ell'$ be the highest level where $L_{\ell'}[i] = \ell'$ or $0$ if 
    level satisfies this condition.\label{newcoren:max-g}\\
    \STATE $\kest(i)\leftarrow\upexpold^{\max\left(\lfloor\frac{\ell' + 1}{4\lceil\log_{1+\psi}(n)\rceil}\rfloor-1, 0\right)}$.
    \STATE $C \leftarrow C \cup \{\kest(i)\}$.\label{newcoren:compute-estimate}
\ENDFOR
\STATE \Return $C$.\label{lkcsr-line:return-noisy-density}
\end{algorithmic}
\end{algorithm}

We first show that using our MAT technique given in~\cref{sec:mat} our algorithm is $\eps$-LEDP.

\begin{lemma}\label{lem:2-approx-ledp}
    \cref{alg:improved-ledp-kcore} is $\eps$-local edge differentially private. 
\end{lemma}

\begin{proof}
    We formulate~\cref{alg:improved-ledp-kcore} in the context of~\cref{alg:local multidimensional Above Threshold} which implies
    that~\cref{alg:improved-ledp-kcore} is $\eps$-LEDP. As in the proof of~\cref{thm:optimal-privacy}, we use an instance of the 
    MAT mechanism with input graph $G$, privacy parameter $\eps$, and threshold vector $\overrightarrow{T}$. We now adaptively define the 
    queries.

    For each round $r$, the $r$-th query consists of $\upexp^{\lfloor r/(\numgrouplevels)\rfloor} - U_i$. The sensitivity of the 
    vector containing such queries for each $v \in V$ is $2$ for each round $r$  (so $\Delta = 2$). Let $a_r(i)$ be the answers to the adaptive 
    queries from round $r$ for node $i$. Then, we can compute $A_i$ (Lines~\ref{lkcsr-line:group-update-1} and~\ref{lkcsr-line:group-update}) 
    (via post-processing) for each node $i$ using only $a_r(i)$ and 
    $L_{\lcur + 1}[i]$ is computed from $A_i$ and inductively from $L_{\lcur}[i]$. The approximate core numbers can be 
    obtained from the $L_{\lcur}[i]$ values from each round 
    via post-processing as shown in the original proof in~\cite{DLRSSY22}.
    Hence, our algorithm can be implemented using MAT and
    by~\cref{thm:mat-ledp} and~\Cref{lem:post} is $\eps$-LEDP.
\end{proof}

\begin{theorem}\label{thm:approx-logn-rounds}
    \cref{alg:improved-ledp-kcore} outputs $\left(2+\eta, O\left(\frac{\log n}{\eps}\right)\right)$-approximate core numbers $\hat{k}(v)$
    with probability at least $1 - O\left(\frac{1}{n^c}\right)$ for any constant $c \geq 1$. 
\end{theorem}

We defer the proof of~\cref{thm:approx-logn-rounds} to~\cref{app:distributed-approx-proof} since it is nearly
identical to the proof of the original algorithm.

Using the MAT technique in~\cref{alg:multidimensional Above Threshold}, we can also remove a factor of $O(\log n)$ in the 
additive error using the $O(\log n)$ algorithm given in
Algorithm 3 of~\cite{DLRSSY22}; however, $\Delta = 2\log n$
for that algorithm when reframed in the MAT framework. 
Thus, we would obtain a $O(\log n)$ round algorithm
which is $\eps$-LEDP but gives $O\left(2 + \eta, 
O\left(\frac{\log^2 n}{\eps}\right)\right)$-approximate
core numbers.
\section{Related Problems}

We now discuss how our improved private $k$-core decomposition algorithm can be used to derive better private algorithms for densest subgraph and low out-degree ordering. We again note that all algorithms given here can be implemented in near-linear time using \Cref{lem:run time} while losing only $1+\eta$ in the multiplicative error; these details are omitted.

\subsection{Differentially Private Densest Subgraph}

We present our $(2, O(\log(n)/\eps))$-approximate edge-differentially private densest subgraph algorithm in~\cref{alg:densest subgraph algorithm}.
Our algorithm follows immediately from our $k$-core decomposition algorithm and the well-known folklore theorem that the density of the 
densest subgraph, $D(G)$, of graph $G$ falls between half of the maximum $k$-core value and the maximum $k$-core value. In fact, one can show
that the $(k_{\max}/2)$-core (where $k_{\max}$ is the maximum value of a non-empty core)
contains the densest subgraph in $G$ (see e.g.\ Corollary 2.2 of~\cite{SLDS23}). 
Furthermore, the $k_{\max}$-core gives a $2$-approximation of the densest subgraph.

We now use~\cref{alg:efficient-algorithm} to obtain differentially private approximations of the $k$-core values of each vertex.
Then, we find the maximum returned core value; we denote this value by $\hat{k}_{\max}$. Finally, we take the induced subgraph 
consisting of all vertices with approximate core values at least $\hat{k}_{\max} - \frac{c'\log(n)}{\eps}$ for sufficiently
large constant $c'$. This means that 
we include all $k_{\max}$ vertices, with high probability, and potentially some vertices with core values at least $k_{\max} - 
\frac{2c'\log(n)}{\eps}$. However, we prove that this still results in a good approximate densest subgraph in~\cref{thm:densest-subgraph}.

\begin{algorithm}[h]
\caption{Densest Subgraph}
\label{alg:densest subgraph algorithm}
\textbf{Input:} Graph $G=(V,E)$, privacy parameter $\eps>0$.\\
\textbf{Output:} An $(2,c'\log(n)/\eps)$-approximate densest subgraph $S\subseteq V$ for sufficiently large constant $c' > 0$\\
\begin{algorithmic}[1]
\STATE Run \Cref{alg:optimal-algorithm} with parameter $\eps$ to obtain approximate core numbers $\hat{k}(v)$ for each $v\in V$
\STATE Let $\hat{k}_{\max} \leftarrow \max_{v \in V}\left(\hat{k}(v)\right)$
\STATE Find $S \leftarrow \{v \mid \hat{k}(v) \geq \hat{k}_{\max} - c'\log(n)/\eps\}$
\STATE Return $S$
\end{algorithmic}
\end{algorithm}

\begin{theorem}\label{thm:densest-subgraph}
    \cref{alg:densest subgraph algorithm} outputs an $\eps$-(local) edge differentially private
    $\left(2, O\left(\frac{\log(n)}{\eps}\right)\right)$-approximate densest subgraph with 
    probability $1 - O\left(\frac{1}{n^c}\right)$ for any constant $c \geq 1$.
\end{theorem}

\begin{proof}
    The privacy of~\cref{alg:densest subgraph algorithm} is guaranteed because we only do postprocessing on the differentially 
    private output of~\cref{alg:optimal-algorithm} and privacy is preserved under postprocessing (\Cref{lem:post}). Then, to prove the 
    additive error, we first observe that with probability at least $1 - O\left(\frac{1}{n^c}\right)$, the vertex set $S$
    contains vertices with core numbers at least $k_{\max} - \frac{c'\log(n)}{\eps}$ for sufficiently large constant $c' > 0$. 
    Then, we can calculate the 
    minimum density of the graph induced by $S$ to be at least

    \begin{align*}
        \frac{|S| \cdot (k_{\max} - 2c'\log(n)/\eps)}{2|S|} = \frac{|S| \cdot k_{max}}{2|S|} - \frac{|S|\cdot 2c'\log(n)}{2\eps|S|} \geq 
        \frac{k_{\max}}{2} - \frac{c'\log(n)}{\eps}.
    \end{align*}
    Since we know $k_{\max}/2 \leq D(G) \leq k_{\max}$, it follows that
    $S$ is a $(2, O(\log(n)/\eps))$-approximate densest subgraph.
\end{proof}

\subsection{Differentially Private Low Out-Degree Ordering}

\cite{DLRSSY22} defines the low out-degree ordering problem to be finding a differentially private ordering where
the out-degree of any vertex is minimized
when edges are oriented from vertices earlier in the ordering to later in the ordering. 
Specifically, there exists an ordering of the vertices of any input graph where
the out-degree is at most the degeneracy $d$ of the graph; the degeneracy of the graph 
is equal to its maximum $k$-core value. The previous best differentially private algorithm for 
the problem gives out-degree at most $(4+\eta)d + O\left(\frac{\log^3 n}{\eps}\right)$.
In this paper, we give a novel algorithm that obtains out-degree bounded by $d + O\left(\frac{\log n}{\eps}\right)$,
with high probability. 

Our algorithm (given in~\cref{alg:low-outdeg} in~\cref{app:low-outdegree-algorithm} due to space constraints) 
is a simple modification of~\cref{alg:optimal-algorithm} where each time we remove a 
vertex $v$ in Line~\ref{line:optimal peeled} of~\cref{alg:optimal-algorithm}, we add $v$ to the end of our current 
ordering. The changes to~\cref{alg:optimal-algorithm} are highlighted in blue.
Thus, our ordering is precisely the order by which vertices are removed. We now prove that 
our algorithm for low out-degree ordering is $\eps$-edge DP and gives an ordering for which the 
out-degree is upper bounded by $d + O\left(\frac{\log n}{\eps}\right)$.

\begin{theorem}
    Our modified~\cref{alg:optimal-algorithm} is $\eps$-(local) edge differentially private and gives a low out-degree ordering
    where the out-degree is upper bounded by $d + O\left(\frac{\log n}{\eps}\right)$.
\end{theorem}

\begin{proof}
    We first prove that our modified algorithm is $\eps$-(local) edge differentially private. First, the proof of~\cref{thm:optimal-utility}
    ensures that the ratio of the probabilities that all vertices $v$ are removed at the same steps in edge-neighboring
    graphs $G$ and $G'$ is upper bounded by $e^{\eps}$. Thus, since the ordering of the vertices follows
    precisely the order by which vertices are removed, then, the ordering that is released is $\eps$-(local) edge DP.

    Now, we prove the upper bound on the out-degree returned by the algorithm. Each removed vertex $v$ has induced degree, when removed,
    that is upper bounded by $k(v) + O\left(\frac{\log n}{\eps}\right)$, with high probability, by~\cref{thm:optimal-utility}. 
    The induced degree when the vertex is removed is exactly equal to its out-degree.
    Then, since $d = \max_{v \in V}\left(k(v)\right)$, the out-degree of the returned 
    ordering is at most $d + O\left(\frac{\log n}{\eps}\right)$, with high probability.
\end{proof}
\appendix
\section{Proof of the Approximation Factor in~\cref{sec:improved-distributed}}\label{app:distributed-approx-proof}

\begin{proof}[Proof of~\cref{thm:approx-logn-rounds}]
    This proof is nearly identical to the proof of Theorem 4.7 in~\cite{DLRSSY22} except
    for two details. First, in~\cref{alg:improved-ledp-kcore}, we use Laplace noise instead
    of symmetric geometric noise and second, we add additional noise to the threshold when determining whether 
    vertices move up levels. Notice that these two changes \emph{only} affects Invariants 1 and 2. Thus, 
    we only need to show that modified versions of Invariants 1 and 2 still hold for our algorithm and 
    Theorem 4.7 and our approximation proof follow immediately.

    \begin{invariant}[Degree Upper Bound]\label{inv:degree-1}
    If node $i \in V_{\lcur}$ (where $V_\lcur$ is the set of nodes in level $\lcur$) and level $\lcur < \numlevels - 1$, then $i$
    has at most $\upexp^{\lfloor r/(\numgrouplevels)\rfloor} + \frac{c\log n}{\eps}$
    neighbors in levels $\geq r$ \whp{} for large enough constant $c > 0$.
    \end{invariant}
    
    \begin{invariant}[Degree Lower Bound]\label{inv:degree-2}
        If node $i \in V_{\lcur}$  (where $V_\lcur$ is the set of nodes in level $\lcur$) and level $\lcur > 0$, then
        $i$ has at least $(1 + \lf)^{\lfloor (r-1)/(\numgrouplevels)\rfloor} - \frac{c \log n}{\eps}$ neighbors in levels $\geq r - 1$ \whp{}
        for large enough constant $c > 0$.
    \end{invariant}

    First, we know that by the density function of the Laplace distribution, any
    noise drawn from the Laplace distributions used in our algorithm
    is upper bounded by $\frac{c\log{n}}{\eps}$ for large enough constant $c > 0$ \whp{}.
    We first show~\cref{inv:degree-1}. A node moves up a level when $\hnup_{i}$ exceeds the threshold given by 
    $\upexp^{\lfloor r/(\numgrouplevels)\rfloor} + \tilde{\ell}(i)$. By our algorithm, $\hnup_{i} = \nup_{i} + X$
    where $X$ and $\tilde{\ell}(i)$ are noises drawn from the Laplace distribution.
    Thus, if a node doesn't move up, then $\nup_{i} \leq \upexp^{\lfloor r/(\numgrouplevels)\rfloor} + \tilde{\ell}(i) - X$.
    The expression on the right is at most $\upexp^{\lfloor r/(\numgrouplevels)\rfloor} + 2 \cdot \frac{c\log{n}}{\eps}$
    \whp for large enough constant $c > 0$. 

    Now, we prove~\cref{inv:degree-2}. If a node is on level $> 0$, then it must have moved up at least one level.
    If it moved up at least one level (from level $r-1$), then its $\hnup_{i}$ at level $r-1$ must be 
    at least $\upexp^{\lfloor (r-1)/(\numgrouplevels)\rfloor} + \tilde{\ell}(i)$. Then, it holds that $\nup_i \geq 
    \upexp^{\lfloor (r-1)/(\numgrouplevels)\rfloor} + \tilde{\ell}(i) - X$ at level $r-1$. The right hand side
    is at least $\upexp^{\lfloor (r-1)/(\numgrouplevels)\rfloor} - 2 \cdot \frac{c\log n}{\eps}$ for large enough constant $c > 0$
    and~\cref{inv:degree-2} also holds.

    We now can use the proof of Theorem 4.7 exactly as written except with the new bounds given in~\cref{inv:degree-1}
    and~\cref{inv:degree-2} to obtain our final approximation guarantee.
\end{proof}

\section{Algorithm for Low Out-Degree Ordering}\label{app:low-outdegree-algorithm}

\begin{algorithm}[h]
\caption{Low Out-Degree Ordering}
\label{alg:low-outdeg}
\textbf{Input:} Graph $G=(V,E)$, privacy parameter $\eps>0$.\\
\textbf{Output:} An $\left(1,O\left(\log(n)/\eps\right)\right)$-approximate $k$-core value of each node $v\in V$\\
\begin{algorithmic}[1]
\STATE $V_0\leftarrow V$, $t\leftarrow 0$, $k=c'\log{n}/\eps$.
\STATE Initialize $\hat{k}(v)\leftarrow 0$ for all $v\in V$
\STATE {\color{blue} Initialize ordered list $J$}
\FOR{$v\in V$}
\STATE $\tilde{\ell}(v)\leftarrow \text{Lap}(4/\eps)$
\ENDFOR{}
\STATE
\WHILE{$k\le n$}
\REPEAT
\STATE $t\leftarrow t+1$, $V_t\leftarrow V_{t-1}$
\FOR{$v\in V_{t-1}$}
\STATE $\nu_t(v)\leftarrow\text{Lap}(8/\eps)$
\IF{$d_{V_{t-1}}(v)+\nu_t(v)\le k+\tilde{\ell}(v)$}
\STATE $V_t\leftarrow V_t-\{v\}$
\STATE {\color{blue} Append $v$ to the end of $J$}
\ENDIF{}
\ENDFOR{}
\UNTIL{$V_{t-1}-V_{t} = \emptyset$}
\STATE
\STATE Update the core numbers $\hat{k}(v)\leftarrow k$ for all nodes $v\in V_t$
\STATE $k\leftarrow k + c'\log{n}/\eps$
\ENDWHILE{}
\STATE {\color{blue} Return $J$}
\end{algorithmic}
\end{algorithm}

\section*{Acknowledgements}
We thank anonymous reviewers for their comments and suggestions, and Jalaj Upadhyay for his feedback on the manuscript. Laxman Dhulipala and George Z. Li are supported by NSF award number CNS-2317194. Quanquan C. Liu is supported by an Apple Simons Research Fellowship.

\bibliographystyle{alpha}
\bibliography{refs}

\end{document}